\documentclass[12pt]{amsart}
\usepackage[utf8]{inputenc}
\usepackage{amssymb, amsmath}
\usepackage{fourier}
\usepackage{xcolor}
\usepackage{enumitem}
\usepackage{fancyvrb}

\usepackage{geometry}
\theoremstyle{plain}

\newtheorem{theorem}{Theorem}

\newtheorem{claim}[theorem]{Claim}

\newtheorem{definition}[theorem]{Definition}

\newtheorem{proposition}[theorem]{Proposition}

\newtheorem{property}[theorem]{Property}

\theoremstyle{remark}
\newtheorem{remark}[theorem]{Remark}

\newtheorem{example}[theorem]{Example}

\newcommand{\bit}{\begin{itemize}}
\newcommand{\eit}{\end{itemize}}
\newcommand{\ben}{\begin{enumerate}}
\newcommand{\een}{\end{enumerate}}
\newcommand{\be}{\begin{equation}}
\newcommand{\ee}{\end{equation}}
\newcommand{\ba}{\begin{array}}
\newcommand{\ea}{\end{array}}

\newcommand{\dd}{\mathrm{d}}

\DeclareMathOperator*{\argmin}{arg\,min}

\newcommand{\abs}[1]{\left|#1\right|}

\newcommand\cA{\mathcal A}
\newcommand\cB{\mathcal B}
\newcommand\cH{\mathcal H}
\newcommand\cP{\mathcal P}

\newcommand\R{\mathbb R}

\renewcommand\l{\lambda}
\newcommand\bh{\cB(\cH)}
\newcommand{\fel}{\frac{1}{2}}
\newcommand{\D}{\mathbf{D}}

\newcommand\tr{\operatorname{Tr}}
\newcommand{\ler}[1]{\left( #1 \right)}

\newcommand\A{\mathcal A}

\newcommand{\lerc}[1]{\left\{ #1 \right\}}

\begin{document}

\title{Quantum Hellinger distances revisited}

\author{J\'ozsef Pitrik}
\address[J\'ozsef Pitrik]{MTA-BME Lend\"ulet (Momentum) Quantum Information Theory Research Group, and Department of Analysis, Institute of Mathematics\\
Budapest University of Technology and Economics\\
H-1521 Budapest, Hungary}
\email{pitrik@math.bme.hu}
\urladdr{http://www.math.bme.hu/\~{}pitrik}

\author{D\'aniel Virosztek}
\address[D\'aniel Virosztek]{Institute of Science and Technology Austria\\
Am Campus 1, 3400 Klosterneuburg, Austria}
\email{daniel.virosztek@ist.ac.at}
\urladdr{http://pub.ist.ac.at/\~{}dviroszt}

\thanks{J. Pitrik was supported by the Hungarian Academy of Sciences Lend\"ulet-Momentum grant for Quantum Information Theory, no. 96 141, and by the Hungarian National Research, Development and Innovation Office (NKFIH) via grants no. K119442, no. K124152, and no. KH129601. D. Virosztek was supported by the ISTFELLOW program of the Institute of Science and Technology Austria (project code IC1027FELL01), by the European Union’s Horizon 2020 research and innovation program under the Marie Sklodowska-Curie Grant Agreement No. 846294, and partially supported by the Hungarian National Research, Development and Innovation Office (NKFIH) via grants no. K124152, and no. KH129601.
}



\begin{abstract}
This short note aims to study quantum Hellinger distances investigated recently by Bhatia et al. \cite{bhatia-paper} with a particular emphasis on barycenters. We introduce the family of generalized quantum Hellinger divergences that are of the form $\phi(A,B)=\mathrm{Tr} \left((1-c)A + c B - A \sigma B \right),$ where $\sigma$ is an arbitrary Kubo-Ando mean, and $c \in (0,1)$ is the weight of $\sigma.$ We note that these divergences belong to the family of maximal quantum $f$-divergences, and hence are jointly convex, and satisfy the data processing inequality (DPI). We derive a characterization of the barycenter of finitely many positive definite operators for these generalized quantum Hellinger divergences. We note that the characterization of the barycenter as the weighted multivariate $1/2$-power mean, that was claimed in \cite{bhatia-paper}, is true in the case of commuting operators, but it is not correct in the general case. 
\end{abstract}
\maketitle

\section{Introduction} \label{sec:intro}
\subsection{Motivation, goals} \label{susec:motiv-goals}
Given a measure space $\ler{X,\A, \mu}$ and probability measures $\rho$ and $\sigma$ that are absolutely continuous with respect to $\mu,$ the classical \emph{squared Hellinger distance} or \emph{Hellinger divergence} of $\rho$ and $\sigma$ is defined as

\be \label{eq:hell-original}
\dd_H^2\ler{\rho, \sigma}=\fel \int_X \ler{\ler{\frac{\dd \rho}{\dd \mu}}^\fel-\ler{\frac{\dd \sigma}{\dd \mu}}^\fel}^2 \dd \mu,
\ee
where $\dd \rho / \dd \mu$ and $\dd \sigma / \dd \mu$ denote the Radon–Nikodym derivatives \cite{hellinger-orig}.
The Hellinger divergence is a special {\it Csisz\'ar-Morimoto $f$-divergence} \cite{csiszar-orig, morimoto-orig} generated by the convex function $f(x)=\ler{\sqrt{x}-1}^2,$ and it has several possible counterparts in quantum information theory. One of them is the squared \emph{Bures distance or Wasserstein metric,} see, e.g., the most recent works of Bhatia et al. \cite{bhat-jain-lim}, Dinh et al. \cite{dinh-et-al-18}, and Moln\'ar \cite{molnar-bures}. Another important quantum analogue of the classical Hellinger divergence has been investigated in \cite{bhatia-paper}, namely the quantity
\be \label{eq:q-hell-def}
\dd_H^2 \ler{A, B}=\tr \ler{\fel\ler{A+B}-A\#B},
\ee
where $A, B$ are density operators representing quantum states, or even more generally, positive operators, and $\#$ is the geometric mean introduced by \emph{Pusz} and \emph{Woronowicz} \cite{pusz-woron}, which is a particularly important \emph{Kubo-Ando mean} \cite{ando-laa-79,kubo-ando,ando-lecture}.
\par
In this note, we introduce a far-reaching generalization of the quantum Hellinger divergence \eqref{eq:q-hell-def}, namely, the family of { \it generalized quantum Hellinger divergences} of the form 
\be \label{eq:g-q-h-d-int}
\phi(A,B)=\mathrm{Tr} \left((1-c)A + c B - A \sigma B \right),
\ee
where $\sigma$ is an arbitrary Kubo-Ando mean, and $c \in (0,1)$ is the weight of $\sigma.$
We will note that these divergences belong to the family of \emph{maximal quantum $f$-divergences,} and hence are jointly convex, and satisfy the data processing inequality (DPI). Moreover, we will show an intimate relation between generalized quantum Hellinger divergences and operator valued Bregman divergences (Claim \ref{claim:breg-conn}). By this close relation, we verify in Claim \ref{claim:div-valid}, that generalized quantum Hellinger divergences are genuine divergences in the sense of \cite[Sec. 1.2 \& 1.3]{amari-book}. Note that this is not the case for maximal quantum $f$-divergences in general, see Remark \ref{rem:nons-m-q-f-div}. As the main result of this paper, we derive a characterization of the barycenter of finitely many positive definite operators for these generalized quantum Hellinger divergences. We will also note that the characterization of the barycenter as the weighted multivariate power mean of order $1/2$, that was claimed in the work of Bhatia et al. \cite[Thm. 9]{bhatia-paper}, is true in the case of commuting operators, but it is not correct in the general case. 

\subsection{Basic notions, notation} \label{susec:b-not-not}
Operator monotone functions mapping the positive half-line $(0, \infty)$ into itself admit a transparent integral-representation by L\"owner's theory.
In the seminal paper of Kubo and Ando \cite{kubo-ando}, the following integral representation was considered:
\be \label{eq:int-rep-orig}
f(x)=\int_{[0, \infty]}\frac{x(1+t)}{x+t} \dd m(t) \qquad \ler{x>0},
\ee
where $m$ is some positive Radon measure on the extended half-line $[0, \infty].$
By a simple push-forward of $m$ by the transformation $T: [0, \infty] \rightarrow [0,1]; \, t \mapsto \l:=\frac{t}{t+1},$ we get the following integral-representation of positive operator monotone functions on $(0, \infty):$
\be \label{eq:int-rep-new}
f_\mu (x)=\int_{[0, 1]}\frac{x}{(1-\l)x+\l} \dd \mu(\l) \qquad \ler{x>0},
\ee
where $\mu=T_{\#} m,$ that is, $\mu(A)=m\ler{T^{-1}(A)}$ for every Borel set $A \subseteq [0,1].$ This representation is also well-known and appears --- among others --- in \cite{hansen-laa-13} and \cite{yam-18}.
Note that if $m$ is absolutely continuous with respect to the Lebesgue measure and $\dd m(t)=\rho(t) \dd t,$ then the density of $\mu=T_{\#} m$ is given by $\dd \mu(\l)=\frac{1}{(1-\l)^2}\rho\ler{\frac{\l}{1-\l}} \dd \l.$
\par
Throughout this note, $\cH$ stands for a finite dimensional complex Hilbert space, $\bh$ denotes the set of all linear operators on $\cH,$ and $\bh^{sa}$ and $\bh^{++}$ stand for the set of all self-adjoint and positive definite operators, respectively. On $\bh^{sa}$ we consider the usual L\"owner order induced by positivity. The \emph{Fr\'echet derivative} of a map $\psi: \bh^{sa} \supseteq \mathcal{U} \rightarrow \mathcal{V}$ at the point $X \in \mathcal{U}$ is denoted by $\D \psi (X)[\cdot].$ Here, $\mathcal{U}$ is an open subset of $\bh^{sa},$ usually the cone of positive definite operators, and the target space $\mathcal{V}$ is usually $\R$ or $\bh^{sa}.$ Note that in the latter case $\D \psi (X)[\cdot]$ is a linear map from $\bh^{sa}$ into itself. The symbol $I$ denotes the identity operator on $\cH.$
\par
For positive definite operators $A, B \in \bh^{++},$ the {\it Kubo-Ando connection} generated by the operator monotone function $f_\mu: (0, \infty) \rightarrow (0, \infty)$ is denoted by $A \sigma_{f_\mu} B,$ and is defined by
\be \label{eq:KA-mean-def}
A \sigma_{f_\mu} B = A^{\fel}f_\mu\ler{A^{-\fel} B A^{-\fel}}A^{\fel}.
\ee
A Kubo-Ando connection $\sigma_{f_\mu}$ is a \emph{mean} if and only if $f_\mu(1)=\mu\ler{[0,1]}=1.$ In the sequel, we will restrict our attention to means. We denote by $\cP\ler{[0,1]}$ the set of all Borel probability measures on $[0,1],$ and by $c\ler{\mu}:=\int_{[0,1]}\l \dd \mu (\l)$ the center of mass of $\mu.$ There is a natural way to assign a weight parameter to a mean $\sigma_{f_\mu},$ namely, $W\ler{\sigma_{f_\mu}}:=f_\mu'(1)=c\ler{\mu}.$ More details about this weight parameter can be found in \cite{yam-18}, we only mention that for the weighted arithmetic, geometric, and harmonic means generated by
$$ a_\l(x)=(1-\l)+\l x, \, g_\l(x)=x^\l, \text{ and } h_\l(x)=\ler{(1-\l) +\l x^{-1}}^{-1},$$
respectively, we have
$W\ler{\sigma_{a_\l}}=W\ler{\sigma_{g_\l}}=W\ler{\sigma_{h_\l}}=\l.$ That is, this weight parameter coincides with the usual one in the most important special cases.

\subsection{Convex order} \label{susec:conv-ord}
The \emph{convex order} is a well-known relation between probability measures; for $\mu, \nu \in \cP\ler{[0,1]},$ we say that $\mu \preccurlyeq \nu$ if for all convex functions $u: [0,1] \rightarrow \R$ we have $\int_{[0,1]} u \, \dd \mu \leq \int_{[0,1]} u \, \dd \nu.$ It is clear that for all $\mu \in \cP\ler{[0,1]}$ with $c\ler{\mu}=\l$ we have $\delta_\l \preccurlyeq \mu \preccurlyeq (1-\l)\delta_0 + \l \delta_1,$ where $\delta_x$ denotes the Dirac mass concentrated on $x.$ For any fixed $x>0,$ the map $\l \mapsto \frac{x}{(1-\l)x+\l}$ is convex. Therefore, if $\mu \preccurlyeq \nu,$ then $f_\mu(x)\leq f_\nu(x)$ for all $x>0,$ and hence $A \sigma_{f_\mu} B \leq A \sigma_{f_\nu} B$ for all $A,B \in \bh^{++}.$ Consequently, if $\nu=\ler{1-c\ler{\mu}}\delta_0 + c\ler{\mu} \delta_1,$ then
$A \sigma_{f_\nu} B-A \sigma_{f_\mu} B$ is always positive, in particular, $\tr \ler{A \sigma_{f_\nu} B-A \sigma_{f_\mu} B}\geq 0.$ This quantity is exactly the one we are interested in.

\section{Basic properties of quantum Hellinger distances} \label{sec:basic-prop}
We are interested in divergences of the form
\be \label{eq:div-def-1}
\phi_\mu(A,B):=\tr \ler{\ler{1-c\ler{\mu}} A + c\ler{\mu} B - A \sigma_{f_\mu} B} ,
\ee
where $\mu \in \cP\ler{[0,1]}.$
To avoid trivialities, we assume in the sequel that the support of $\mu$ is strictly larger than $\{0,1\},$ and therefore, $f_\mu$ is non-affine --- in fact, it is strictly concave
\par
If $\mu$ is the \emph{arcsine distribution,} that is, $\dd \mu(\l)=\frac{1}{\pi \sqrt{\l(1-\l)}} \dd \l,$ then
$$
\phi_\mu(A,B)=\tr \ler{\frac{1}{2}(A + B) - A \# B},
$$
where $\#$ is the Pusz-Woronowitz geometric mean \cite{pusz-woron}. The square root of this quantity (up to an irrelevant multiplicative constant) was considered in \cite{bhatia-paper} as a possible quantum (or matrix) version of the classical Hellinger distance. Therefore, we will call the quantities of the form \eqref{eq:div-def-1} \emph{generalized quantum Hellinger divergences.}
\par
We easily get that
\be \label{eq:phi-manip}
\phi_\mu(A,B)=\tr \lerc{A \cdot g_\mu \ler{A^{-\fel}BA^{-\fel}}},
\ee
where $g_\mu: (0, \infty) \rightarrow [0, \infty)$ is defined by
\be \label{eq:g-def}
g_\mu(x)=\ler{1-c\ler{\mu}}+c\ler{\mu} x-f_\mu(x).
\ee
\begin{remark} \label{rem:nons-m-q-f-div}
We note that $g_\mu$ is operator convex as $f_\mu$ is operator concave, and hence generalized quantum Hellinger divergences belong to the family of \emph{maximal quantum $f$-divergences} studied for example in \cite{hiai-mosonyi-17,hiai-neumann,matsumoto-chapter,petz-ruskai}. This latter divergence class consists of quantities of the form $S_f(A,B)=\tr A f \ler{A^{-\fel}BA^{-\fel}},$ where $A, B \in \bh^{++},$ and $f: \, (0, \infty) \rightarrow \R$ is operator convex \cite{hiai-mosonyi-17, petz-ruskai}. However, this level of generality may lead to counter-intuitive phenomena. For instance, the maximal quantum $f$-divergence can be negative (see, e.g., \cite[Example 4.4]{hiai-mosonyi-17}, where $f(x)=x \log{x},$ and $S_f\ler{I, e^{-1} I}=- \mathrm{dim}\ler{\cH} e^{-1}<0$); and it may happen that $S_f(A,A)>0$ for all $A \in \bh^{++}$ (see, e.g., \cite[Example 4.2]{hiai-mosonyi-17}, where $f(x)=x^2,$ and $S_f(A,A)=\tr A>0$ for all $A \in \bh^{++}$). That is, maximal quantum $f$-divergences are not divergences in the sense of \cite[Sec. 1.2 \& 1.3]{amari-book} in general. In particular, they are not necessarily positive definite. (We call a divergence $D$ positive definite, if $D(A,B)\geq 0$ for every $A,B \in \bh^{++},$ and $D(A,B)=0$ if and only if $A=B.$)
\end{remark}

Now we check that generalized quantum Hellinger divergences are intimately related to \emph{operator valued Bregman divergences,} and hence are reasonable measures of dissimilarity and genuine divergences in the sense of \cite[Sec. 1.2 \& 1.3]{amari-book}.

\subsection{The relation with Bregman divergences} \label{subsec:rel-breg}
Note that $h_\mu:=-f_\mu$ is an operator convex function, and that
$$
g_\mu(x)=\ler{1-c\ler{\mu}}+c\ler{\mu} x+h_\mu(x)=h_\mu(x)-h_\mu(1)-h_\mu'(1)(x-1).
$$
The operator valued Bregman divergence generated by the operator convex function $h_\mu$ reads as follows:
$$
H_{h_\mu}^{(op)}(X,Y)=h_\mu(X)-h_\mu(Y)-\D h_\mu(Y)[X-Y].
$$
In particular,
$$
H_{h_\mu}^{(op)}\ler{A^{-\fel}BA^{-\fel},I}=h_\mu(A^{-\fel}BA^{-\fel})-h_\mu(I)-\D h_\mu(I)\left[A^{-\fel}BA^{-\fel}-I \right].
$$
As $\D h_\mu (I)$ coincides with the multiplication by the constant $-c\ler{\mu},$ and $h_\mu'(I)=-c\ler{\mu} I,$
we get that
$$
H_{h_\mu}^{(op)}\ler{A^{-\fel}BA^{-\fel},I}=g_{\mu}(A^{-\fel}BA^{-\fel}).
$$
Therefore, we obtain the following claim.
\begin{claim} \label{claim:breg-conn}
The generalized quantum Hellinger divergence $\phi_\mu$ defined in \eqref{eq:div-def-1} can be expressed by an operator valued Bregman divergence as follows:
\be \label{eq:breg-relation}
\phi_\mu(A,B)=\tr \lerc{A \cdot H_{h_\mu}^{(op)}\ler{A^{-\fel}BA^{-\fel},I}} \qquad \ler{A, B \in \bh^{sa}}.
\ee
\end{claim}
For a detailed study of Bregman divergences on matrices we refer to \cite{pv-15}.
\par
Now we are in the position to check that generalized quantum Hellinger divergences are genuine divergences in the sense of Amari \cite[Sec. 1.2 \& 1.3]{amari-book}.

\begin{claim} \label{claim:div-valid}
For any $\mu \in \cP{[0,1]},$ the map
\be \label{eq:phi-mu-map}
\phi_\mu: \, ^{++} \times \bh^{++} \rightarrow [0,\infty); \, (A,B) \mapsto \phi_\mu(A,B)
\ee
satisfies the followings.
\ben[label=(\roman*)]
\item
$\phi_\mu(A,B)\geq 0$ and $\phi_\mu(A,B)=0$ if and only if $A=B.$
\item
The first derivative of $\phi_\mu$ in the second variable vanishes at the diagonal, that is,
$\D\ler{\phi_\mu(A,\cdot)}(A)=0 \in \mathrm{Lin}\ler{\bh^{sa},\R}$ for all $A\in \bh^{++}.$
\item 
The second derivative of $\Phi_\mu$ in the second variable is positive at the diagonal, that is,
$\D^2 \ler{\phi_\mu(A,\cdot)}(A)[Y,Y] \geq 0$ for all $Y \in \bh^{sa}.$
\een
\end{claim}

\begin{proof}
Bregman divergences are clearly divergences (see, e.g., \cite[Sec. 1]{bhatia-paper}).That is,
\ben[label=(\roman*)]
\item
$H_{h_\mu}^{(op)}\ler{A^{-\fel}BA^{-\fel},I}\geq 0 \in \bh,$ and $H_{h_\mu}^{(op)}\ler{A^{-\fel}BA^{-\fel},I}=0$ if and only if $A=B,$
\item
$\D\ler{H_{h_\mu}^{(op)}\ler{A^{-\fel} \, \cdot \, A^{-\fel},I}}(A)=0 \in \mathrm{Lin}\ler{\bh^{sa}}$ for every $A\in \bh^{++},$
\item 
$\D^2 \ler{H_{h_\mu}^{(op)}\ler{A^{-\fel} \, \cdot \, A^{-\fel},I}}(A)[Y,Y] \geq 0 \in \bh$ for all $Y \in \bh^{sa}.$
\een
Now Claim \ref{claim:div-valid} follows from Claim \ref{claim:breg-conn}.
\end{proof}

\subsection{Joint convexity, data processing inequality} \label{subsec:convexity-DPI}
As generalized quantum Hellinger divergences belong to the family of maximal quantum $f$-divergences, they are jointly convex and they satisfy the data processing inequality, which is particularly important from the quantum information theory viewpoint. For details, see \cite{hiai-mosonyi-17,hiai-neumann,matsumoto-chapter,petz-ruskai}. We recall these important properties for convenience.

\begin{property}[Joint convexity] \label{prop:hell-div-j-conv}
The generalized quantum Hellinger divergence $\phi_\mu$ defined in \eqref{eq:div-def-1} is jointly convex on $\bh^{++} \times \bh^{++}.$
\end{property}

\begin{property}[Data processing inequality] \label{prop:DPI}
Let $T: \bh \rightarrow \bh$ be a \emph{quantum channel,} that is, a completely positive and trace preserving (CPTP) map. Let $\mu \in \cP{[0,1]}$ be arbitrary. Then
\be \label{eq:DPI}
\phi_\mu\ler{T(A),T(B)} \leq \phi_\mu \ler{A,B}
\ee
holds for every $A, B \in \bh^{++}.$
\end{property} 

\section{Barycenters} \label{sec:barycenters}

The notion of barycenter (or least squares mean) plays a central role in averaging procedures related to various topics in mathematics and mathematical physics. Given a metric space $\ler{X, \rho}$ and an $m$-tuple $a_1, \dots, a_m$ in $X$ with positive weights $w_1, \dots, w_m$ such that $\sum_{j=1}^m w_j=1,$ the barycenter (or Fr\'echet mean or Karcher mean or Cartan mean) is defined to be
$$
\argmin_{x \in X}\sum_{j=1}^m w_j \rho^2\ler{a_j,x}.
$$
In our setting, $X=\bh^{++},$ and the generalised quantum Hellinger divergence $\phi_\mu$ plays the role of the squared distance $\rho^2,$ although it is not the square of any true metric in general.
\par
That is, we consider the optimization problem
\be \label{eq:opt-prob}
\argmin_{X \in \bh^{++}} \sum_{j=1}^m w_j \phi_\mu\ler{A_j, X},
\ee
where the positive definite operators $A_1, \dots, A_m$ and the weights $w_1, \dots w_m$ are fixed. By the strict concavity of $f_\mu,$ the function
$$
X \mapsto \phi_\mu \ler{A,X}=\tr \ler{\ler{1-c\ler{\mu}} A + c\ler{\mu} X - A^{\fel}f_\mu\ler{A^{-\fel} X A^{-\fel}}A^{\fel}}
$$
is strictly convex on $\bh^{++},$ see, e.g., \cite[2.10. Thm.]{carlen}.
Therefore, there is a unique solution $X_0$ of \eqref{eq:opt-prob}, and it is necessarily a critical point of the function $X \mapsto \sum_{j=1}^m w_j \phi_\mu\ler{A_j, X}.$ That is, it satisfies
\be \label{eq:diff-vanishes}
\D \ler{\sum_{j=1}^m w_j \phi_\mu\ler{A_j, \cdot}}(X_0)[Y]=0 \qquad \ler{Y \in \bh^{sa}}.
\ee

Easy computations give that
\be \label{eq:diff-transformed}
\D \ler{\sum_{j=1}^m w_j \phi_\mu\ler{A_j, \cdot}}(X)[Y]= c\ler{\mu} \tr Y -\sum_{j=1}^m w_j \tr \D F_{\mu, A_j}(X)[Y],
\ee
where for a positive definite operator $A,$ the map $F_{\mu,A}: \bh^{++} \rightarrow \bh^{++}$ is defined by
\be \label{eg:f-a-def}
F_{\mu,A}(X):=A \sigma_{f_\mu} X = A^{\fel}f_\mu\ler{A^{-\fel} X A^{-\fel}}A^{\fel}.
\ee

By differentiating \eqref{eq:int-rep-new}, we have
\be \label{eq:div-x-t-rep}
\D f_\mu(X) [Y]=\int_{[0,1]} \l \ler{(1-\l) X + \l I}^{-1} Y \ler{(1-\l) X + \l I}^{-1} \dd \mu (\l)
\ee
for $X \in \bh^{++}, \, Y \in \bh^{sa}.$
Consequently,
$$
\D F_{\mu, A_j}(X)[Y]
$$
$$
=\int_{[0,1]} \l A_j^\fel \ler{(1-\l) A_j^{-\fel} X A_j^{-\fel}+\l I }^{-1} A_j^{-\fel} Y A_j^{-\fel} \ler{(1-\l) A_j^{-\fel} X A_j^{-\fel}+\l I }^{-1} A_j^\fel \dd \mu (\l)
$$
\be \label{eq:diff-concrete-form}
=\int_{[0,1]} \l \ler{(1-\l)X A_j^{-1}+ \l I}^{-1} Y \ler{(1-\l) A_j^{-1} X+ \l I}^{-1}  \dd \mu (\l).
\ee

By the linearity and the cyclic property of the trace, we get from \eqref{eq:diff-transformed} and \eqref{eq:diff-concrete-form} that \eqref{eq:diff-vanishes} is equivalent to
\be \label{eq:}
\tr \left[Y \ler{c\ler{\mu} I-\sum_{j=1}^m w_j \int_{[0,1]} \l \abs{(1-\l) A_j^{-1} X+ \l I}^{-2} \dd \mu (\l)}\right] = 0 \qquad \ler{Y \in \bh^{sa}},
\ee
where $|\cdot|$ stands for the absolute value of an operator, that is, $\abs{Z}=\ler{Z^*Z}^{\fel}.$ 
This latter equation amounts to 
\be \label{eq:barycenter-necessary}
c\ler{\mu} I= \sum_{j=1}^m w_j \int_{[0,1]} \l \abs{(1-\l) A_j^{-1} X+ \l I}^{-2} \dd \mu (\l).
\ee
So we obtained the following characterization of the barycenter.

\begin{theorem} \label{thm:fo}
Let $\mu \in \cP{[0,1]}$ and let $\phi_\mu$ be the generalized quantum Hellinger divergence generated by $\mu,$ that is,
$$
\phi_\mu\ler{A,B}=\tr \ler{\ler{1-c\ler{\mu}} A + c\ler{\mu} B - A \sigma_{f_\mu} B} \qquad \ler{A,B \in \bh^{++}}.
$$
Then the barycenter (or Cartan mean or Fr\'echet mean or Karcher mean) of the positive definite operators $A_1, \dots, A_m$ with positive weights $w_1, \dots, w_m$ with respect to $\phi_\mu,$ i.e.,
$$
\argmin_{X \in \bh^{++}} \sum_{j=1}^m w_j \phi_\mu\ler{A_j, X}
$$
coincides with the unique positive definite solution of the matrix equation
\be \label{eq:barycenter-char}
c\ler{\mu} I= \sum_{j=1}^m w_j \int_{[0,1]} \l \abs{(1-\l) A_j^{-1} X+ \l I}^{-2} \dd \mu (\l).
\ee
\end{theorem}

\section{The commutative case}
In this section we show that in the commutative case formula \eqref{eq:barycenter-char} can be greatly simplified (see \eqref{eq:comm-crit-point} later), furthermore, the conditions on $f$ can be relaxed. Recall that in the general non-commutative case, the generating function $f$ was operator monotone (or equivalently, operator concave), and hence smooth ($C^\infty$), see \eqref{eq:int-rep-orig} and \eqref{eq:int-rep-new}. When dealing with commuting operators, we need concavity only in the classical one-variable sense, and hence we require much less regularity on $f.$ For now, we only require that $f: (0, \infty) \rightarrow \R$ is a strictly concave $C^1$ function.
\par
Let $\cA \subset \bh$ be a maximal Abelian subalgebra (MASA). In this commutative case, the proper analogue of the generalized quantum Hellinger divergence \eqref{eq:div-def-1} is
$$
\phi_f(A,B):=\tr \ler{\ler{f(1)-f'(1)} A + f'(1) B - A^\fel f\ler{A^{-\fel} B A^{-\fel}} A^{\fel} }
$$
\be \label{eq:div-def-comm}
=\tr \ler{\ler{f(1)-f'(1)} A + f'(1) B - A f\ler{A^{-1}B} } \qquad \ler{A,B \in \cA \cap \bh^{++}}.
\ee
Note that now there is no underlying measure involved and the function class that we choose the $f'$s from is much larger than that in the general non-commutative case.
Also note that
\be \label{eq:phi-manip-comm}
\phi_f(A,B)=\tr A \cdot g\ler{A^{-\fel}B A^{-\fel}}=\tr A \cdot g\ler{A^{-1}B}=\tr A \cdot g\ler{B^\fel A^{-1} B^\fel},
\ee
where $g(x)=f(1)+f'(1)(x-1)-f(x).$
We easily get that for $A, X \in \cA \cap \bh^{++}$ and $Y \in \cA \cap \bh^{sa}$ we have
\be \label{eq:comm-diff-form} 
\D \ler{\phi_f(A,\cdot)}(X)[Y]=\tr \ler{f'(I)-f' \ler{X^\fel A^{-1} X^\fel}}Y=\tr \ler{f'(I)-f' \ler{A^{-1}X}}Y,
\ee
and therefore,
$$
\D \ler{\sum_{j=1}^m w_j\phi_f\ler{A_j,\cdot}}(X)[Y]=\tr \ler{f'(I)-\sum_{j=1}^m w_j f' \ler{X^\fel A_j^{-1} X^\fel}}Y
$$
\be \label{eq:comm-diff-form-sum}
=\tr \ler{f'(I)-\sum_{j=1}^m w_j f' \ler{A_j^{-1}X}}Y.
\ee
That is, the derivative $\D \ler{\sum_{j=1}^m w_j\phi_f\ler{A_j,\cdot}}(X)$ vanishes if and only if
\be \label{eq:deriv-vanish}
\sum_{j=1}^m w_j f' \ler{A_j^{-1}X}=\sum_{j=1}^m w_j f' \ler{X^\fel A_j^{-1} X^\fel}=f'(I),
\ee
or equivalently,
\be \label{eq:deriv-vanish-2}
X=\frac{1}{f'(1)}\sum_{j=1}^m w_j X f' \ler{A_j^{-1}X}=\frac{1}{f'(1)} \sum_{j=1}^m w_j X^\fel f'\ler{\ler{X^{-\fel} A_j X^{-\fel}}^{-1}} X^\fel.
\ee
We obtained the following
\begin{proposition} \label{propo:comm-bary}
The critical point of the function $X \mapsto \sum_{j=1}^m w_j\phi_f\ler{A_j,X}$ is the unique solution $X \in \cA \cap \bh^{++}$ of the equation
\be \label{eq:comm-crit-point}
X=\frac{1}{f'(1)}\sum_{j=1}^m w_j X f' \ler{A_j^{-1}X}=\frac{1}{f'(1)} \sum_{j=1}^m w_j X^\fel f'\ler{\ler{X^{-\fel} A_j X^{-\fel}}^{-1}} X^\fel.
\ee
\end{proposition}

So in the commutative case, the equation characterizing the barycenter \eqref{eq:comm-crit-point} is simpler than that in the non-commutative case \eqref{eq:barycenter-char}. Note that if all the $A_j$'s are in the same MASA $\A \subset \cB(\cH)$, then the barycenter is also in $\A,$ and hence it has the form described in Proposition \ref{propo:comm-bary}. One way to show this is to use the data processing inequality (DPI) for the orthogonal projection onto $\A$ which is completely positive and trace preserving, and which is denoted by $\mathbf{E}_\A$ to express the analogy with the classical conditional expectation. So let $X_0$ be the unique minimizer of $X \mapsto \sum_{j=1}^m w_j \phi_\mu \ler{A_j,X}.$ Now
$$
\sum_{j=1}^m w_j \phi_\mu \ler{A_j,\mathbf{E}_{\A}\ler{X_0}}
=
\sum_{j=1}^m w_j \phi_\mu \ler{\mathbf{E}_{\A}\ler{A_j},\mathbf{E}_{\A}\ler{X_0}}
\leq
\sum_{j=1}^m w_j \phi_\mu \ler{A_j,X_0},
$$
hence $\mathbf{E}_{\A}\ler{X_0}=X_0$ which means that $X_0 \in \A.$
We also note that under the assumption $A_j X= X A_j$ for all $j'$s, \eqref{eq:barycenter-char} clearly coincides with \eqref{eq:comm-crit-point}, because $c\ler{\mu}=f_\mu'(1),$ and in this case, by the identity
\be \label{eq:f-mu-diff}
f_\mu'(x)=\frac{\dd}{\dd x} \ler{\int_{[0, 1]}\frac{x}{(1-\l)x+\l} \dd \mu(\l)}=\int_{[0, 1]}\frac{\l}{\ler{(1-\l)x+\l}^2} \dd \mu(\l)
\ee
we have
$$
\sum_{j=1}^m w_j \int_{[0,1]} \l \abs{(1-\l) A_j^{-1} X+ \l I}^{-2} \dd \mu (\l)
=\sum_{j=1}^m w_j \int_{[0,1]} \l \ler{(1-\l) A_j^{-1} X+ \l I}^{-2} \dd \mu (\l)
$$
\be \label{eq:acwf}
=\sum_{j=1}^m w_j f_\mu'\ler{A_j^{-1} X}.
\ee

\begin{example} \label{ex:power}
Let $f_t(x)=x^t$ for $t \in (0,1).$ Then $\phi_{f_t}$ is of the form
\be \label{eq:comm-phi-ex-power-fn}
\phi_{f_t}(A,B)=\tr \ler{(1-t)A+t B -A^{1-t}B^t}=\tr \ler{(1-t)A+t B -A \#_t B},
\ee
and the barycenter equation \eqref{eq:comm-crit-point} reads as
\be \label{eq:comm-pow-bary}
X=\sum_{j=1}^m w_j A_j^{1-t} X^t=\sum_{j=1}^m w_j A_j \#_t X=\sum_{j=1}^m w_j X \#_{1-t} A_j.
\ee
That is, the barycenter coincides with the {\it weighted power mean of order} $1-t,$ which is by definition the unique positive definite solution of the equation $X=\sum_{j=1}^m w_j X \#_{1-t} A_j,$ see \cite[Def. 3.2]{lim-palfia-jfa}.
This example does not contain new results, the above characterization of the barycenter as weighted power mean can be found, e.g., in \cite{amari-paper} or in \cite{schwander-nielsen}.
\end{example}

\begin{remark}
By the special choice $t=1/2$ in Example \ref{ex:power}, we get that the claim of Bhatia et al. saying that the barycenter and the weighted power mean of order $1/2$ coincide \cite[Thm. 9]{bhatia-paper} is true in the commutative case.
\end{remark}

\begin{example}\label{ex:log}
Set $f(x)=\log{x}.$ Then $\phi_f$ is the relative entropy, that is,
\be \label{eq:comm-phi-ex-log}
\phi_{\log}(A,B)=\tr \ler{A \ler{\log{A}-\log{B}}+B-A},
\ee
and the barycenter equation \eqref{eq:comm-crit-point} reads as
\be \label{eq:comm-log-bary}
X=\sum_{j=1}^m w_j X \ler{A_j^{-1} X}^{-1}=\sum_{j=1}^m w_j A_j .
\ee
That is, the barycenter coincides with the weighted sum of the $A_j$'s. This is well-known, see, e.g., the remarks after Theorem 4 in \cite{bhatia-paper}.
\end{example}
Note that we get Example \ref{ex:log} from Example \ref{ex:power} if we take the limit $t \to 0.$
Indeed, 
\be \label{eq:t-to-null-limit}
\lim_{t \to 0}\frac{1}{t}\phi_{f_t}(A,B)=\lim_{t \to 0}\frac{1}{t} \tr A \cdot g_t \ler{A^{-1}B},
\ee
where $g_t(x)=1+t(x-1)-x^t,$ and $\lim_{t \to 0} \frac{1}{t}\ler{1-x^t}=\log{x}$ in the locally uniform topology.


\section{Remarks} \label{sec:rems}
\subsection{A note on a paper of Bhatia et al}
In our view, Theorem 9 in \cite{bhatia-paper} is not true in general. The proof contains a gap, namely, using their notation, the fact that $I$ is a critical point for $g$ does not imply that $X_0$ is a critical point for $f,$ although formula (54) in \cite{bhatia-paper} is correct.
\par
It is true, that for commuting operators, \eqref{eq:barycenter-char} and \eqref{eq:comm-crit-point} coincide. However, these equations are different without the assumption of commutativity. To demonstrate the difference, we take the following example. Let $\mu$ be the arcsine distribution, $\dd \mu(\l)=\frac{1}{\pi \sqrt{\l(1-\l)}} \dd \l,$ let $m=2, w_1=w_2=\fel,$ and 
$$
A_1:=\left[\ba{cc} 4 & 0 \\ 0 & 1 \ea\right], \qquad A_2:=4 \left[\ba{cc} 1/2 & 1/2   \\ 1/2 & 1/2 \ea\right]+1 \left[\ba{cc} 1/2 & -1/2 \\ -1/2 & 1/2 \ea\right]=\fel \left[\ba{cc} 5 & 3 \\ 3 & 5 \ea\right].
$$
Then numerical optimization performed by Wolfram Mathematica \cite{hellinger-numerics} shows that
$$
\hat X_0:=\argmin_{X \in \bh^{++}} \sum_{j=1}^2 \fel \phi_\mu\ler{A_j, X}
$$
\be \label{eq:argmin-num}
=\argmin_{X \in \bh^{++}} \fel \tr \ler{\fel \ler{A_1+A_2}+ X- \ler{ A_1 \# X + A_2 \# X}}
=\left[\ba{cc} 2.99035 & 0.634419 \\ 0.634419 & 1.72151 \ea\right].
\ee

Note that both $A_1$ and $A_2$ have real entries. Therefore, $A_j \# \overline{X}=\overline{A_j \# X},$ and hence $\phi_\mu\ler{A_j,X}=\phi_\mu\ler{A_j,\overline X}$ holds for every $X \in \bh^{++}$ and $j \in \{1,2\},$ where $\overline X$ denotes the entrywise complex conjugate of $X.$  Consequently, the strict convexity of the functions $X \mapsto \phi_\mu\ler{A_j,X}, \, j \in \{1,2\}$ implies that $\argmin_{X \in \bh^{++}} \sum_{j=1}^2 \fel \phi_\mu\ler{A_j, X}$ has real entries. So it is enough to minimize numerically over the cone of positive definite $2 \times 2$ matrices with real entries \cite{hellinger-numerics}.
\par
However, the barycenter obtained numerically in \eqref{eq:argmin-num} does not coincide with the weighted power mean of order $1/2$ as
\be \label{eq:num-noteq}
\fel \ler{ A_1 \# \hat X_0 + A_2 \# \hat X_0}=\left[\ba{cc} 3.02915 & 0.673215 \\ 0.673215 & 1.68272 \ea\right] \neq \hat X_0.
\ee
Note that after the publication of our manuscript on \emph{arXiv.org}, a correction of \cite{bhatia-paper} dedicated to this problem was released \cite{bhatia-corr}.

\subsection{A possible measure of non-commutativity}
Motivated by the observations above, we introduce a function that quantifies the noncommutativity of the positive definite operators $A_1, \dots, A_m.$
\begin{definition} \label{def:nc-meas}
Given $\mathbf{A}=\ler{A_1, \dots, A_m} \in \ler{\bh^{++}}^m, \, \mathbf{w}=\ler{w_1, \dots, w_m} \in (0,1]^m$ with $\sum_{j=1}^m w_j=1, \, \, \mu \in \cP\ler{[0,1]},$ and a convenient metric $\rho$ on $\bh^{++},$ the $\ler{\mathbf{w},\mu, \rho }$-dependent measure of the non-commutativity of $A_1, \dots, A_m$ is defined as
\be \label{eq:nc-meas-def}
\rho \ler{\mathrm{BC}\ler{\mathbf{A}, \mathbf{w},\mu}, \mathrm{M}\ler{\mathbf{A}, \mathbf{w},\mu}}
\ee
where
$$
\mathrm{BC}\ler{\mathbf{A}, \mathbf{w},\mu}=\argmin_{X \in \bh^{++}} \sum_{j=1}^m w_j \phi_\mu\ler{A_j, X},
$$
i.e., $\mathrm{BC}\ler{\mathbf{A}, \mathbf{w},\mu}$ is the solution of \eqref{eq:barycenter-char}, and $\mathrm{M}\ler{\mathbf{A}, \mathbf{w},\mu}$ is a $\mu$-dependent $\mathbf{w}$-weighted mean of $A_1, \dots, A_m$ defined as the unique solution of the matrix equation \eqref{eq:comm-crit-point} that we recall here for convenience:
$$
X=\frac{1}{c\ler{\mu}}\sum_{j=1}^m w_j X^{\fel}f_\mu'\ler{\ler{X^{-\fel} A_j X^{-\fel}}^{-1}} X^{\fel}.
$$
\end{definition}
The detailed study of the quantity \eqref{eq:nc-meas-def} is beyond the scope of this paper, however, it may be the subject of subsequent works. 

\subsubsection*{Acknowledgements}
We are grateful to Mil\'an Mosonyi for drawing our attention to Ref.'s \cite{bhatia-paper, hiai-mosonyi-17, hiai-neumann, matsumoto-chapter, mosonyi-ogawa, petz-ruskai}, for comments on earlier versions of this paper, and for several discussions on the topic. We are also grateful to Mikl\'os P\'alfia for several discussions; to L\'aszl\'o Erd\H{o}s for his essential suggestions on the structure and highlights of this paper, and for his comments on earlier versions; and to the anonymous referee for his/her valuable comments and suggestions.

\bibliographystyle{amsplain}

\end{document}